\documentclass[letterpaper, 10 pt, conference]{ieeeconf}  

\usepackage{mathrsfs}
\usepackage{amssymb}
\usepackage{mathtools}
\usepackage{amsmath}
\usepackage{nccmath}
\usepackage{graphicx}
\usepackage{caption}
\usepackage{pifont}
\usepackage{cuted}
\usepackage{comment}
\usepackage{cite}
\usepackage{xcolor}
\usepackage{comment}
\usepackage{lipsum}  

\usepackage{algorithm}
\usepackage{algpseudocode}

\usepackage{caption}
\usepackage{subcaption}

\newcommand{\argminF}{\mathop{\mathrm{argmin}}\limits} 

\newtheorem{lemma}{Lemma}

\newtheorem{definition}{Definition}
\newtheorem{remark}{Remark}
\newtheorem{theorem}{Theorem}
\newtheorem{assumption}{Assumption}

\usepackage{color}

\IEEEoverridecommandlockouts                              

\title{\LARGE \bf
Robust Control Barrier Functions for Sampled-Data Systems
}

\author{Pradeep Sharma Oruganti, Parinaz Naghizadeh and Qadeer Ahmed%
\thanks{Pradeep Sharma Oruganti and Qadeer Ahmed are affiliated with the Center for Automotive Research at The Ohio State University. Parinaz Naghizadeh is with the Department of Electrical and Computer Engineering at the University of California,  San Diego. Emails: \{oruganti.6, ahmed.358\}@osu.edu; parinaz@ucsd.edu.}
}

\begin{document}

\maketitle
\thispagestyle{empty}
\pagestyle{empty}

\begin{abstract}

This paper studies the problem of safe control of sampled-data systems under bounded disturbance and measurement errors with piecewise-constant controllers. To achieve this, we first propose the High-Order Doubly Robust Control Barrier Function (HO-DRCBF) for continuous-time systems where the safety enforcing constraint is of relative degree 1 or higher. We then extend this formulation to sampled-data systems with piecewise-constant controllers by bounding the evolution of the system state over the sampling period given a state estimate at the beginning of the sampling period. We demonstrate the proposed approach on a kinematic obstacle avoidance problem for wheeled robots using a unicycle model. We verify that with the proposed approach, the system does not violate the safety constraints in the presence of bounded disturbance and measurement errors. 
\end{abstract}

\section{Introduction}
\label{sec: intro}

There has been an increasing incorporation of connectivity and automation in systems that have been traditionally isolated and mechanical. Autonomous systems in transportation, manufacturing and healthcare are being tasked to perform several safety-critical functions. As such, assurances on their \emph{safe} functioning are expected before widespread implementations, demanding that the safety properties be encoded in the controller design.

In this context, a frequently used notion of safety is that of \emph{set invariance} where the system trajectory never leaves a subset of its state-space (termed the \emph{safe set}) \cite{blanchini1999set}. Control Barrier Functions (CBFs) have grown in popularity for real-time safety critical control given their efficient Quadratic Programming (QP) formulation and their ability to render the safe set forward invariant \cite{ames2019control}. They have been successfully showcased on robotic \cite{cortez2019control}, automotive \cite{ames2016control}, and machine learning applications \cite{dean2021guaranteeing}. 

Some variations of CBFs have also been proposed in the literature. For instance, for many mechanical systems, the actuation happens at the acceleration level (for example, wheel force or joint torque inputs) while the safety requirements are position dependent, making these constraints of \emph{relative degree} two (formally defined in Section \ref{sec: prelims}) \cite{cortez2019control}. To handle such constraints of relative degree greater than 1, high-order variations of CBFs have been proposed \cite{xiao2021high}. Additionally, \emph{robust} variations of CBFs considering worst-case disturbance bounds have been proposed in \cite{jankovic2018robust}. State measurement noise has been addressed in the stochastic perspective in \cite{clark2019control}, and using measurement estimate error bounding in \cite{dean2021guaranteeing}; the latter was recently extended to high-order safety constraints in \cite{oruganti2023safe}. 

While these CBF variations are theoretically sound, they may face an additional challenge during practical implementation: in practice, systems are controlled digitally, with state measurements and control inputs being constant over each sampling time periods \cite{breeden2021control}. Considering this discrepancy, similar to the continuous time variations, discrete-time CBFs \cite{agrawal2017discrete, xiong2022discrete}, CBFs for sampled-data systems \cite{cortez2019control, usevitch2021adversarial, singletary2020control}, and approaches based on approximate discrete-time models \cite{taylor2022safety}, have been developed. However, \emph{robust} variations of CBFs for sampled-data systems remain unexplored.

In this paper, we extend the CBF for sampled-data systems proposed in \cite{breeden2021control} to their robust variation. Specifically, we propose the \emph{Doubly Robust Control Barrier Function} (DRCBF), which guarantees safety of sampled-data systems with piecewise-constant controllers under both (bounded) disturbance noise and (bounded) state measurement errors (hence, \emph{doubly} robust). Our approach is applicable to systems with relative degree 1 as well as high-order safety constraints. We implement our proposed approach on an obstacle avoidance problem for a wheeled robot modeled using the unicycle model, and verify that it can respect safety constraints under both disturbance and measurement errors. Additionally, we reduce conservative behavior of the proposed safe controller by incorporating interval reachability techniques. We further move beyond the existing robust CBF literature by considering a special mis-matched disturbance case where the relative degree of the CBF with respect to the disturbance is one lower than that of the controller. 

The rest of the paper is organized as follows. Section \ref{sec: prelims} presents the background and problem formulation. Our main results are presented in Section \ref{sec: main_results}, and showcased on a numerical example in Section \ref{sec: numerical}. All proofs are provided in the appendix.
\section{Background and Problem Formulation}
\label{sec: prelims}

\paragraph{Preliminaries} A continuous function $\alpha: [0, a] \rightarrow [0, \infty]$ is said to belong to class $\mathcal{K}$ if it is strictly increasing and $\alpha(0)=0$ \cite{khalil2015nonlinear}. Given functions $p(\mathbf{x})$ and $q(\mathbf{x})$, $L_qp(\mathbf{x}) := \nabla p(\mathbf{x}) \cdot q(\mathbf{x})$ is called the \emph{Lie derivative} of $p$ along $q$. The boundary and interior of a closed set $\mathcal{S} \subset \mathbb{R}^n$ are represented by $\partial \mathcal{S}$ and $\textrm{Int}(\mathcal{S})$, respectively. 

\paragraph{Safety and CBFs} Consider the continuous system:
\begin{equation}
    \label{eq: cont_sys}
    \dot{\mathbf{x}} = f\big(\mathbf{x}\big) + g\big(\mathbf{x}\big)\mathbf{u},
\end{equation}
where $\mathbf{x} \in \mathcal{D} \subset \mathbb{R}^n$ represents the system state, $\mathbf{u} \in \mathcal{U} \subset \mathbb{R}^q$ is the control input, with {the set of admissible inputs} $\mathcal{U}$ being a compact set, and $f: \mathbb{R}^n \rightarrow \mathbb{R}^n$, $g: \mathbb{R}^n \rightarrow \mathbb{R}^{n \times {q}}$, are locally Lipschitz continuous functions. 
We use \emph{forward invariance} to formalize the notion of safety for this system.
\begin{definition}[Forward invariant set \cite{xiao2021high}]
    A set of states $\mathcal{S}$ is said to be rendered forward invariant by a controller $u: \mathcal{D} \rightarrow \mathcal{U} $ if starting at $\mathbf{x}(t_0) \in \mathcal{S}$, $\mathbf{x}(t) \in \mathcal{S}$ for all $t > t_0$.
\end{definition}

In this paper, we will use the terms ``safety'' and ``forward invariance'' interchangeably. If the set $\mathcal{S}$ is rendered forward invariant for system \eqref{eq: cont_sys} by some feedback control $\mathbf{u}(t) = k(\mathbf{x}, t)$, locally Lipschitz in $\mathbf{x}$ and piecewise continuous in $t$, then the system is said to be \emph{safe} with respect to $\mathcal{S}$ and is referred to as the \emph{safe set}. We assume that the safe set $\mathcal{S}$ can be represented as the 0-superlevel set of a continuously differentiable function $h: \mathbb{R}^n \rightarrow \mathbb{R}$, i.e., 
\begin{equation}
\label{eq: safe_set}
\mathcal{S} := \{\mathbf{x} \in \mathcal{D} \subset \mathbb{R}^n: h(\mathbf{x}) \geq 0 \}~.
\end{equation}

{Control Barrier Functions (defined below) have emerged as a popular tool for rendering \eqref{eq: safe_set} forward invariant for \eqref{eq: cont_sys}.
\begin{definition}[Control Barrier Function \cite{ames2016control}]
\label{def: cbf} Given a set $\mathcal{S}$ as in \eqref{eq: safe_set}, $h(\mathbf{x})$ is a Control Barrier Function (CBF) if there exists a class $\mathcal{K}$ function $\alpha$ such that,
\begin{equation}
\label{eq: cbf-def}
    \sup_{\mathbf{u} \in \mathcal{U}} \Big[ L_fh(\mathbf{x}) + L_gh(\mathbf{x})\mathbf{u} + \alpha(h(\mathbf{x})) \Big] > 0, ~ \forall \mathbf{x} \in \mathcal{S}.
\end{equation}
\end{definition}} 

\begin{theorem}[\cite{ames2019control}] Given a CBF $h(\mathbf{x})$ and set $\mathcal{S}$ defined as \eqref{eq: safe_set}, any Lipschitz continuous controller $\mathbf{u} \in \mathcal{U}$ that satisfies \eqref{eq: cbf-def} renders the set $\mathcal{S}$ forward invariant for system (\ref{eq: cont_sys}).
\end{theorem}

\paragraph{Sampled-data systems} While these conditions have been introduced for continuous time systems, in practice, the systems are run by sampling the system state at fixed (or varying) discrete time steps $T = t_{k+1} - t_k$ and using Zero-Order Hold (ZOH) based digital controllers, i.e., $\mathbf{x}(t)=\mathbf{x}_k$, $\mathbf{u}(t) = \mathbf{u}_k, ~\forall t \in [t_k, t_{k+1})$. In such scenarios, the following result from \cite{cortez2019control} ensures the forward invariance of $\mathcal{S}$.
\begin{theorem}[\cite{cortez2019control}, Thm.~2]
For a safe set $\mathcal{S}$ defined as \eqref{eq: safe_set}, assume that the functions $L_fh$, $L_gh$, and $\alpha \circ h$ are Lipschitz continuous with Lipschitz constants $l_{L_fh}$, $l_{L_gh}$, and $l_{\alpha \circ h} \in \mathbb{R}_{\geq 0}$. At the sampled state $\mathbf{x}_k$, a piecewise-constant control input $\mathbf{u}_k \in \mathcal{U}$ satisfying: 
\begin{equation}
    \sup_{\mathbf{u}_k \in \mathcal{U}} \Big[ L_fh(\mathbf{x}_k) + L_gh(\mathbf{x}_k)\mathbf{u}_k + \alpha(h(\mathbf{x}_k)) - \frac{l_1}{l_2}(e^{l_2 T} - 1)\Big] \geq 0 
\end{equation}
where $l_1 = (l_{L_fh} + l_{L_gh} u_{\textrm{max}} + l_{\alpha \circ h}) \sup_{\mathbf{x} \in \mathcal{D}, \mathbf{u}_k \in \mathcal{U}} \lVert f(\mathbf{x}) + g(\mathbf{x})\mathbf{u}_k \rVert$ and $l_2 = l_{L_fh} + l_{L_gh} u_{\textrm{max}}$ ensures that $\mathcal{S}$ is forward invariant.
\end{theorem}

\paragraph{High order control barrier functions} As noted in Section \ref{sec: intro}, there exist many practical scenarios where the \emph{relative degree} (defined below) of the safety constraints are some $m > 1$, i.e., $L_gL_f^{r}h(\mathbf{x}) = 0,~\forall r=\{0, 1, \ldots,m-1\}$. 
\begin{definition} [Relative degree \cite{khalil2015nonlinear}]
The relative degree of a continuously differentiable function $h: \mathbb{R}^n \rightarrow \mathbb{R}$ with respect to system (\ref{eq: cont_sys}) is the number of times, $m$, we need to differentiate $h$ until $\mathbf{u}$ shows up explicitly in $h^{(m)}$. 
\end{definition}

High-Order Control Barrier Functions (HOCBF) have been proposed to handle such scenarios.  
\begin{definition}[HOCBF \cite{xiao2021high}]
\label{def: ho_cbf}
    Consider the system \eqref{eq: cont_sys}, and $\mathcal{S}$ defined as \eqref{eq: safe_set} with $h$ being a sufficiently smooth continuous function of relative degree $m$. Let $\{\mathcal{S}_i\}_{i = 1}^m$ be a collection of sets of the form $\mathcal{S}_i := \{\mathbf{x} \in \mathcal{D} : \psi_{i-1}(\mathbf{x}) \geq 0\}$, where $\psi_0 := h(\mathbf{x})$ and $\psi_i := \dot{\psi}_{i-1}(\mathbf{x}) + \alpha_i(\psi_{i-1}(\mathbf{x}))$, for all $i \in \{1, \ldots , r\}$. Then, $h$ is a high-order control barrier function (HOCBF) on $\mathcal{D} \supset \cap_{i=1}^m \mathcal{S}_i$ if there exists a collection of class-$\mathcal{K}$ functions $\{\alpha\}_{i=1}^m$ such that for all $\mathbf{x} \in \mathcal{D}$, $\sup_{u \in \mathcal{U}} \psi_r \geq 0$.
\end{definition}

With this definition, the following theorem proves the forward invariance of $\cap_{i=1}^m \mathcal{S}_i$:
\begin{theorem}[\cite{xiao2021high}]
\label{th: ho_cbf}
    If $h$ is a HOCBF on $\mathcal{D} \subset \mathbb{R}^n$, then any locally Lipschitz continuous controller $\mathbf{u} \in \mathcal{U}$ such that $\psi_r \geq 0$ renders $\cap_{i=1}^m \mathcal{S}_i$ forward invariant for system \eqref{eq: cont_sys}.
\end{theorem}

\subsection{Problem formulation}
Consider the following continuous system:
\begin{equation}
    \label{eq: d_cont_sys}
    \dot{\mathbf{x}} = f\big(\mathbf{x}\big) + g\big(\mathbf{x}\big)\mathbf{u} + p(\mathbf{x})\mathbf{d}~,
\end{equation}
where $\mathbf{d} \in \mathcal{P} \subset \mathbb{R}^v$ is some external disturbance, and $p: \mathbb{R}^n \rightarrow \mathbb{R}^{n \times v}$ is locally Lipschitz continuous. We make the following assumption about the disturbance.
\begin{assumption}
\label{asmp: bounded_disturbance}
The set $\mathcal{P}$ is compact and there exists a $\gamma > 0$ such that $ \max_{\mathbf{d} \in \mathcal{P}} \lVert \mathbf{d} \rVert \leq \gamma$. 
\end{assumption}

Assumption \ref{asmp: bounded_disturbance} is typical in practical systems, and imposes a bounded disturbance $\mathbf{d}$ on the system \cite{cortez2019control}. {The \emph{Robust-CBF} was proposed in \cite{jankovic2018robust} to ensure safety of continuous-time systems with disturbance bounded by a known constant}. Additionally, assuming matched disturbance, i.e., the disturbance relative degree is the same as the input relative degree, \cite{cohen2022high, tan2021high} proposed variations of the HOCBF that are robust to disturbance. 

In this paper, in addition to similarly considering the disturbance $\mathbf{d}$, we  account for an uncertainty in the sensor measurements that are used to estimate the system state. Specifically, given a sensor measurement $\mathbf{z} \in \mathbb{R}^l$, we assume that {the controller} has access to an imperfect estimate $\hat{\mathbf{x}} := \mathbf{x} + \mathbf{e}(\mathbf{x})$ of the true state $\mathbf{x}$, where $\mathbf{e} \in \mathcal{E}(\mathbf{z})$ represents the measurement errors. 
\begin{assumption}
\label{asmp: bounded_meas_error}
$\mathcal{E}(\mathbf{z})$ is compact, and there exists $\epsilon(\mathbf{z})$ such that $\max_{\mathbf{e} \in \mathcal{E}(\mathbf{z})} \lVert \mathbf{e} \rVert \leq \epsilon(\mathbf{z})$. 
\end{assumption}

Assumption \ref{asmp: bounded_meas_error} implies bounded errors in state measurement. With the estimated states, the evolution of the closed-loop system is given by:
\begin{equation}
\label{eq: cl_estim_sys}
    \dot{\mathbf{x}} = f(\mathbf{x}) + g(\mathbf{x})k(\mathbf{z}, \hat{\mathbf{x}}_k) + p(\mathbf{x})\mathbf{d}.
\end{equation} 
In this paper, given a safe-set $\mathcal{S}$ defined as \eqref{eq: safe_set} with a sufficiently smooth continuous function $h$, and a constant sampling time-period $T$, we look to design a closed-loop, piecewise-constant (ZOH) controller $\mathbf{u}_k = k(\mathbf{z}, \hat{\mathbf{x}}_k), ~\forall t \in [t_k, t_k+1)$, that ensures the forward invariance of \eqref{eq: cl_estim_sys} with respect to $\mathcal{S}$. {We do this for cases where the relative degree of $h$ with respect to the input is some $m \geq 1$}. 
\section{Main Results}
\label{sec: main_results}

In this section, we propose the High Order Doubly Robust Control Barrier Functions (HO-DRCBF), i.e., robust to both external disturbance and measurement errors, for systems whose relative degree with respect to $h$ is some $m \geq 1$. We first propose the HO-DRCBF for continuous-time systems and then extend this formulation to the sampled-data systems with piece-wise constant controllers. Note that we obtain the Doubly Robust CBF (DRCBF) when the relative degree of the systems with respect to $h$ is 1 by simply applying $m = 1$ in the proposed results (see appendix for detailed discussion). 

Following Definition \ref{def: ho_cbf}, we observe that the system may have differing Disturbance Relative Degree ($\textrm{DRD}$) and Input Relative Degree ($\textrm{IRD}$) with respect to $h$. We first look at the case when $\textrm{DRD} = \textrm{IRD} = m$ (also known as the matched disturbance case). This assumption implies that the external disturbance does not appear before the control input when taking higher order derivatives of $h$. A second case, when this uncertainty appears at one derivative of $h$ lower than the derivative where the control input appears, i.e., $\textrm{IRD} - \textrm{DRD} = 1$, is briefly discussed in Section \ref{sec: mismatached_disturbance}. 

\subsection{For continuous time systems}
For this case, we make an additional assumption:

\begin{assumption}
\label{asmp: lips_conditions} 
The function $\psi_{m-1}$ is locally Lipschitz continuous. 
\end{assumption}

The above assumption helps in generating a lower bound in the HOCBF constraint to ensure invariance considering disturbance and measurement errors. With this assumption, we now introduce the HO-DRCBF:

\begin{definition}[HO-DRCBF]
\label{def: ho_dmr_cbf}
Consider the closed-loop system \eqref{eq: cl_estim_sys}, and a collection of sets $\{\mathcal{S}\}_{i=1}^m$ with functions $\psi_0 := h(\mathbf{x})$ and $\{\psi_i\}_{i=1}^m$ defined as in Definition \ref{def: ho_cbf}. If Assumptions \ref{asmp: bounded_disturbance}, \ref{asmp: bounded_meas_error}, and \ref{asmp: lips_conditions} hold, the sufficiently smooth function $h: \mathbb{R}^n \rightarrow \mathbb{R}$ is a \emph{high-order doubly robust CBF} if there exists a control input $\mathbf{u} \in \mathcal{U}$ such that:
\begin{multline}
\label{eq: ho_dmr_cbf_condition}
    \sup_{\mathbf{u} \in \mathcal{U}} \Big[ L_f\psi_{m-1}(\hat{\mathbf{x}}) + L_g\psi_{m-1}(\hat{\mathbf{x}})\mathbf{u} 
    - (a_i(\mathbf{z}) + b(\mathbf{z})\lVert \mathbf{u} \rVert) \\ + \alpha_{m-1}(\psi_{m-1}(\hat{\mathbf{x}}))
    - \lVert L_p\psi_{m-1}(\mathbf{x}) \rVert \gamma \Big] \geq 0~. 
\end{multline}
\end{definition}

Following Definition~\ref{def: ho_dmr_cbf}, denote the set of control inputs satisfying \eqref{eq: ho_dmr_cbf_condition} by $K_{\textrm{ho-dr}}$. Our first result below shows that any Lipschitz continuous controller $\mathbf{u}^* \in K_{\textrm{ho-dr}}$ renders the safe set $\mathcal{S}$ forward invariant in the presence of disturbance $\mathbf{d}$ and measurement noise $\mathbf{e}$. 

\begin{lemma}
\label{lemma: ho_dmr_cbf_ct}
Consider the collection of sets $\{\mathcal{S}\}_{i=1}^m$ with functions $\psi_0 := h(\mathbf{x})$ and $\{\psi_i\}_{i=1}^m$ defined as in Definition \ref{def: ho_cbf}. Given Assumptions \ref{asmp: bounded_disturbance}, \ref{asmp: bounded_meas_error}, and \ref{asmp: lips_conditions} hold, and if $h$ is a HO-DRCBF for system (\ref{eq: cl_estim_sys}) with parameter functions $(a(\mathbf{z}), b(\mathbf{z})) := (\epsilon(\mathbf{z})(l_{L_f\psi_{r-1}} + l_{\alpha_{r-1} \circ \psi_{r-1}}),\epsilon(\mathbf{z})(l_{L_g\psi_{r-1}}))$, where $l_{L_f\psi_{r-1}}$, $l_{\alpha_{r-1} \circ \psi_{r-1}}$, and $l_{L_g\psi_{r-1}}$ are the Lipschitz constants of $L_f\psi_{r-1}$, $\alpha_{r-1} \circ \psi_{r-1}$, and $L_g\psi_{r-1}$ respectively, then any locally Lipschitz continuous controller $\mathbf{u} \in K_{\textrm{ho-dr}}$ renders $\cap_{i=1}^m \mathcal{S}_i$ forward invariant.
\end{lemma}

Observe that when there are no disturbances ($\mathbf{d} = \mathbf{0}$) and measurement errors ($\epsilon(\mathbf{z}) = 0$), we recover the original HOCBF condition \cite{xiao2021high}.  Compared to the HO-MR-CBF proposed in \cite{oruganti2023safe}, we observe that an additional (negative) disturbance term $\lVert L_ph \rVert \gamma$ shows up in our inequality. Notice that when compared to the scenario without the two sources of uncertainties, it can be seen that the controller now has to act conservatively to ensure that the left-hand side of \eqref{eq: ho_dmr_cbf_condition} is not negative (i.e., system state remains in $\mathcal{S}$). The safe control input can be obtained from the following optimization:
\begin{equation}
\label{eq: ho_dmr_cbf_optim}
    \begin{split}
    & K_{\text{safe}}(\hat{\mathbf{x}}) = \argminF_{\mathbf{u} \in \mathcal{U}} \frac{1}{2} \lVert \mathbf{u} - K_{\text{perf}}(\hat{\mathbf{x}}) \rVert^2 \\
    & \text{  s.t.  } L_f\psi_{m-1}(\hat{\mathbf{x}}) + L_g\psi_{m-1}(\hat{\mathbf{x}})\mathbf{u}
    - (a_i(\mathbf{z}) + b(\mathbf{z})\lVert \mathbf{u} \rVert) \\ 
    & + \alpha_{m-1}(\psi_{m-1}(\hat{\mathbf{x}})) - \lVert L_p\psi_{m-1}(\mathbf{x}) \rVert \gamma \geq 0
    \end{split} 
\end{equation}
where $K_{\text{perf}}(\hat{\mathbf{x}}_k)$ is some performance-oriented (potentially \emph{safety-agnostic}) control input. While this is not in the typical Quadratic Program (QP) formulation, it can be converted into a Second-Order Conic Program (SOCP) (see \cite[Section B]{dean2021guaranteeing}).

\subsection{For sampled-data systems}
Next, we extend the HO-DRCBF to sampled-data systems. The main motivation comes from the fact that since the control input $\mathbf{u}_k$ applied at time $t_k$ is constant between sampling times, safety needs to be ensured throughout the time-step. Formally, given a sampled state-estimate $\hat{\mathbf{x}}_k \in \mathcal{S}$ at $t_k$, we aim to show that $\mathcal{R}(\hat{\mathbf{x}}_k, T) \in \mathcal{S}$, where $\mathcal{R}(\hat{\mathbf{x}}_k, T)$ represents all the \emph{actual} states reachable from $\hat{\mathbf{x}}_k$ in times $t \in [kT, (k+1)T)$ \cite{gurriet2019realizable}. Before we introduce the result for sampled-data systems, we introduce the following lemma which bounds the \emph{true} system states $\mathbf{x}(t)$ given the sampled state estimate $\hat{\mathbf{x}}_k$ over time $t \in [kT, (k+1)T)$ in presence of disturbances satisfying Assumption \ref{asmp: bounded_disturbance}. 

\begin{lemma}
\label{lemma: state_bounds}
For all $\hat{\mathbf{x}}_k \in \mathcal{D}$ and $t \in [kT, (k+1)T)$, the closed loop trajectories of \eqref{eq: d_cont_sys} satisfy
\begin{equation*}
    \lVert \mathbf{x}(t) - \hat{\mathbf{x}}_k \rVert \leq \epsilon(\mathbf{z}) + T (\Delta + \lVert L_ph(\mathbf{x}) \rVert \gamma)~,
\end{equation*}
where $\Delta = \sup_{\mathbf{x} \in \mathcal{D}, \mathbf{u}_k \in \mathcal{U}} \lVert f(\mathbf{x}) + g(\mathbf{x})\mathbf{u}_k \rVert$.
\end{lemma}
With this bound, we formally propose the high-order doubly robust CBF for sampled-data systems below. For readability, let $v(\mathbf{z}) := \epsilon(\mathbf{z}) + T(\Delta + \lVert L_ph(\mathbf{x})) \rVert \gamma)$. 
\begin{definition}[HO-SD-DRCBF]
Consider the closed-loop system \eqref{eq: cl_estim_sys}, and a collection of sets $\{\mathcal{S}\}_{i=1}^m$ with functions $\psi_0 := h(\mathbf{x})$ and $\{\psi_i\}_{i=1}^m$ defined as in Definition \ref{def: ho_cbf}. The sufficiently smooth function $h: \mathbb{R}^n \rightarrow \mathbb{R}$ is a \emph{high-order doubly robust} for the sampled data system \eqref{eq: cl_estim_sys}, if there exists suitable class-$\mathcal{K}$ functions $\{\alpha\}_{i=1}^m$ such that there exists a control input $\mathbf{u} \in \mathcal{U}$ such that: 
\begin{multline}
\label{eq: ho_sd_dmr_cbf_condition}
    \sup_{\mathbf{u}_k \in \mathcal{U}} \Big[ L_f\psi_{m-1}(\hat{\mathbf{x}}_k) + L_g\psi_{m-1}(\hat{\mathbf{x}}_k)\mathbf{u}_k + \alpha_{m-1}(\psi_{m-1}(\hat{\mathbf{x}}_k))  \\ 
    - (l_{L_f\psi_{m-1}} + l_{L_g\psi_{m-1}}\lVert \mathbf{u}_k \rVert + l_{\alpha \circ \psi_{m-1}}) v(\mathbf{z}) \\  
    - \lVert L_p\psi_{m-1}(\mathbf{x}) \rVert \gamma \Big] \geq 0~.
\end{multline}
\end{definition}
Similar to the previous section, denote the set of control inputs that satisfy \eqref{eq: ho_sd_dmr_cbf_condition} as $K_{\textrm{ho-sd-drcbf}}$.
\begin{lemma}
\label{lemma: ho_dmr_cbf_sd}
Consider the collection of sets $\{\mathcal{S}\}_{i=1}^m$ with functions $\psi_0 := h(\mathbf{x})$ and $\{\psi_i\}_{i=1}^m$ defined as in Definition \ref{def: ho_cbf}. Given Assumptions \ref{asmp: bounded_disturbance}, \ref{asmp: bounded_meas_error}, and \ref{asmp: lips_conditions} hold, and if $h$ is a HO-SD-MR-CBF for system (\ref{eq: cl_estim_sys}), then any locally Lipschitz continuous controller $\mathbf{u}_k \in K_{\textrm{ho-sd-drcbf}}$ renders $\cap_{i=1}^m \mathcal{S}_i$ forward invariant.
\end{lemma}

Similar to Lemma \ref{lemma: ho_dmr_cbf_ct}, we see that robust margins are added to the original CBF constraint. The difference here is the additional term $T(\Delta + \lVert L_ph(\mathbf{x})) \rVert \gamma)$ to bound the evolution of the true system state $\mathbf{x}$ over the entire sampling period. Additionally, similar to \eqref{eq: ho_dmr_cbf_optim}, we can formulate an optimization problem with the constraint being \eqref{eq: ho_sd_dmr_cbf_condition} to find the safe control input for the sampled-data system. 

\begin{remark}[Reducing conservative behavior]
To further improve feasibility and reduce conservative behavior, we can estimate the worst-case change in safety over the sampling period, given an approximation of $\mathcal{R}(\hat{\mathbf{x}}, T)$ (similar to \cite{breeden2021control}). For example, following the parametric approach proposed in \cite{xiao2021high}, under matched disturbance, consider the following margin function:
\begin{multline}
    \mathtt{margin}^{\mathtt{mat}} := (L_f\psi_{m-1}(y) - L_f\psi_{m-1}(\hat{\mathbf{x}}_k)) + \\ 
    (L_g\psi_{m-1}(y) - L_g\psi_{m-1}(\hat{\mathbf{x}}_k))\mathbf{u}_k + L_p\psi_{m-1}(y)\mathbf{d} \\ 
    + (p_m\psi_{m-1}^{q_m}(y) - p_m\psi_{m-1}^{q_m}(\hat{\mathbf{x}}_k))~.
\end{multline}
We now define the set $K_{\mathtt{reach}}^{\mathtt{mat}}$ as:
\begin{multline*}
    K_{\mathtt{reach}}^{\mathtt{mat}} := \{ \mathbf{u}_k: \sup_{\mathbf{u}_k \in \mathcal{U}} \Big[ L_f\psi_{m-1}(\hat{\mathbf{x}}_k) + L_g\psi_{m-1}(\hat{\mathbf{x}}_k)\mathbf{u}_k \\ 
    + \alpha(\psi_{m-1}(\hat{\mathbf{x}}_k)) \Big] - \sup_{y \in \mathcal{R}(\hat{\mathbf{x}}, T), \mathbf{u}_k \in \mathcal{U}, \mathbf{d} \in \mathcal{P}}\mathtt{margin}^{\mathtt{mat}} \geq 0\}~.
\end{multline*}
Similar to the proof of Lemma \ref{lemma: ho_dmr_cbf_sd}, it can be shown that, under the stated assumptions, $\mathbf{u}_k \in K_{\mathtt{reach}}^{\mathtt{mat}}$ indeed renders $\cap_{i=1}^m \mathcal{S}_i$ forward invariant. We use the concept of \emph{interval reachability} to obtain an over-approximation of $\mathcal{R}(\hat{\mathbf{x}}, T)$ \cite{meyer2019tira}. 

To further achieve better (less conservative) behavior, assuming that the disturbance bounds are accurate, the only design choice available for the controllers is the choice of the class $\mathcal{K}$ function. A potential approach to obtain the optimal design parameters for our proposed robust HOCBF for sampled data systems is to design an iterative algorithm similar to the paramterization method proposed in \cite{xiao2021high}. Further reduction in conservatism could also be obtained by incorporating an MPC type approach, i.e., remedying the myopic nature of our proposed controller by giving it some foresight. This would in-turn demand building \emph{robust discrete-time CBFs}, similar to \cite{cosner2023robust}. When the uncertainty bounds used are conservative approximations of the actual quantities, an adaptive or learning based approach could be utilized to estimate and update the CBF \cite{cohen2022high, lopez2023unmatched}.
\end{remark}

\subsection{Mis-matched disturbance}
\label{sec: mismatached_disturbance}

Lastly, we look at a special mis-matched disturbance case when $\textrm{IRD} - \textrm{DRD} = 1$. Here, we assume that the disturbance \emph{does not} affect the control inputs. This is a practical assumption as the \emph{actual} system inputs $\mathbf{u}$ can be augmented as states of the system and can be thought of as integrals of some virtual inputs $\mathbf{v}$, i.e., $\dot{\mathbf{u}} = \mathbf{v}$. Following the parametric approach, we obtain the following formulation for $\psi_m$:
\begin{equation*}
    \begin{split}
        & \psi_{m} := L_f^mh(\mathbf{x}) + L_gL_f^{m-1}h(\mathbf{x})\mathbf{u} + L_pL_f^{m-1}h(\mathbf{x})\mathbf{d} \\ 
        & +  L_f[L_pL_f^{m-2}h(\mathbf{x})\mathbf{d}] + L_p[L_pL_f^{m-2}h(\mathbf{x})\mathbf{d}]\mathbf{d} \\ 
        & + O(h(\mathbf{x})) + p_{m} \psi_{m-1}^{q_m}~,
    \end{split}
\end{equation*}
where $(p_i, q_i), ~i \in \{1,2,\ldots, m\}$ with $q_i \geq 1$ are parameters of the HOCBF, $O(h(\mathbf{x}))$ are functions involving $L_fh(\mathbf{x})$ and its higher derivatives. 
Note that since we assume that the disturbance affects only the states and not the control inputs (maybe virtual) of the system, we have $L_g[L_pL_f^{m-2}h(\mathbf{x})\mathbf{d}]\mathbf{u} = 0$. Now, similar to previous sections, to prove the forward invariance of $\cap_{i=1}^m \mathcal{S}_i$, we look to show $\psi_m \geq 0$ with the robust variation obtained by taking the norms of all terms involving $\mathbf{d}$. Observe that in this case the constraint is still affine in $\mathbf{u}$, conserving the resulting QP formulation of the resulting the safety controller. 

\section{Numerical Example}
\label{sec: numerical}

In this section, we showcase our proposed approach on the safe control of a wheeled robot with states $(x, y, \theta, v)$, corresponding to the robot's x-position, y-position, heading, and velocity, respectively. The robot is required to avoid an obstacle in presence of disturbance noise and measurement errors. The safety requirement is given by $h(\mathbf{x}) = (x-x_o)^2 + (y-y_o)^2 - D^2 \geq 0$, where $(x_o, y_o)$ are the coordinates of the obstacle and $D = 5\textrm{m}$ is the safe distance to be maintained from $(x_o, y_o)$. We model the robot using the unicycle model given by:
\begin{equation}
\label{eq: robot}
    \begin{bmatrix}
    \dot{x} \\
    \dot{y} \\
    \dot{\theta} \\
    \dot{v} 
    \end{bmatrix} = \begin{bmatrix}
    v \cos{\theta} \\
    v \sin{\theta} \\
    0 \\
    0 
    \end{bmatrix} + \begin{bmatrix}
    0 & 0\\
    0 & 0\\
    1 & 0\\
    0 & 1
    \end{bmatrix}\begin{bmatrix}
    u_1 \\
    u_2
    \end{bmatrix} + \begin{bmatrix}
    1 & 0 \\
    0 & 1 \\
    0 & 0 \\
    0 & 0 
    \end{bmatrix}
    \begin{bmatrix}
    d_1 \\
    d_2 
    \end{bmatrix}~.
\end{equation}
We assume a disturbance $\mathbf{d} = [d_1, d_2]$, with $d_1, d_2 \in [-0.3, 0.3]$, in the evolution of the system position variables, a measurement error of $[-0.5, 0.5]$ meters in the position estimates, and perfect measurement of heading and velocity. The control limits are set to $[-1, -2]^T \leq \mathbf{u} \leq [1, 2]^T$. 

The state of the robot is sampled every $T=0.1\text{s}$. We obtain the next state by numerically integrating the dynamics using the Runge-Kutta 4\textsuperscript{th} order approach with constant control input (ZOH approach). Observe that the $\textrm{DRD}$ of the system with respect to $h(\mathbf{x})$ is 1, whereas the $\textrm{IRD}$ is 2. We use the TIRA reachability toolbox \cite{meyer2019tira} to obtain $\mathcal{R}(\hat{\mathbf{x}}, T)$\footnote{We note that the reachable sets obtained through the toolbox are the reachable sets \emph{at} time instant $t = (k+1)T$, whereas for the current approach we require a representation of all the states reachable \emph{through} time $t = kT$ to $t=(k+1)T$. A discussion on this is provided in the appendix, where we numerically observe that for the current system, the margins estimated from the reachable set at time $(k+1)T$ from time $kT$ are larger than the margins at any time $t \in [kT, (k+1)T)$.}. The safe control input is obtained by solving the following optimization problem:
\begin{equation}
\label{eq: robot_optim}
    \begin{split}
    & K_{\text{safe}}(\hat{\mathbf{x}}) = \argminF_{\mathbf{u} \in \mathcal{U}} \frac{1}{2} \lVert \mathbf{u} - K_{\text{perf}}(\hat{\mathbf{x}}) \rVert^2 \\
    & \text{  s.t.  } f_v(\hat{\mathbf{x}}, \mathbf{u}) - \sup_{\mathbf{x} \in \mathcal{R}(\hat{\mathbf{x}}, T), \mathbf{u}_k \in \mathcal{U}, \mathbf{d} \in \mathcal{P}} \mathtt{margin} \geq 0~,
    \end{split} 
\end{equation}
where $\mathtt{margin} = f_v(\mathbf{x}, \mathbf{u}) - f_v(\hat{\mathbf{x}}, \mathbf{u}) + f_d(\mathbf{x}, \mathbf{d})$,  and 
\begin{equation*}
    \begin{split}
        f_v(\mathbf{x}, \mathbf{u}) & = L^2_fh(\mathbf{x}) + 2L_fh(\mathbf{x}) + 2h(\mathbf{x}) + L_gL_fh(\mathbf{x})\mathbf{u} \\
        f_d(\mathbf{x}, \mathbf{d}) & = L_pL_fh(\mathbf{x})\mathbf{d} + L_f[L_ph(\mathbf{x})\mathbf{d}] + L_p[L_ph(\mathbf{x})\mathbf{d}]\mathbf{d} + \\ & 2L_ph(\mathbf{x})\mathbf{d}~.
    \end{split}
\end{equation*}

\begin{figure}[t]
    \centering
    \includegraphics[width=0.75\columnwidth]{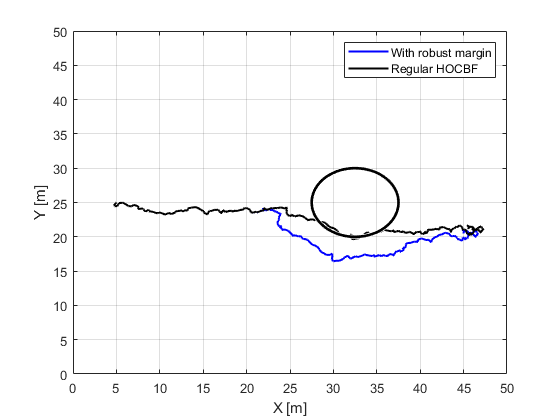}
    \caption{Robot trajectory obtained by using the vanilla HOCBF approach (black curve) and the proposed robust variation (blue curve).}
    \label{fig: robot_trajs}
\end{figure}

\begin{figure}[t]
    \centering
    \includegraphics[width=0.75\columnwidth]{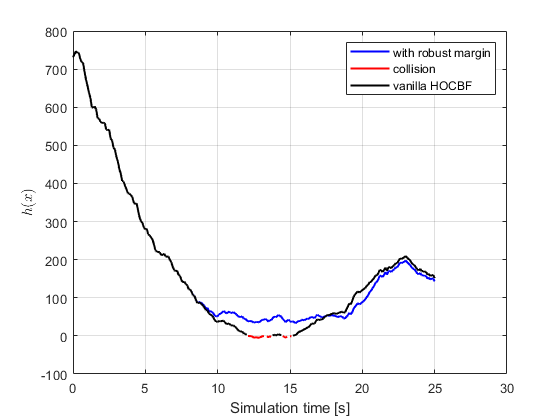}
    \caption{Value of $h(\mathbf{x})$ obtained by using the vanilla HOCBF approach (black curve) and the proposed robust variation (blue curve). The red part of the vanilla HOCBF trajectory indicates a collision with the obstacle.}
    \label{fig: h_x}
\end{figure}

\begin{figure}[t]
    \centering
    \includegraphics[width=0.75\columnwidth]{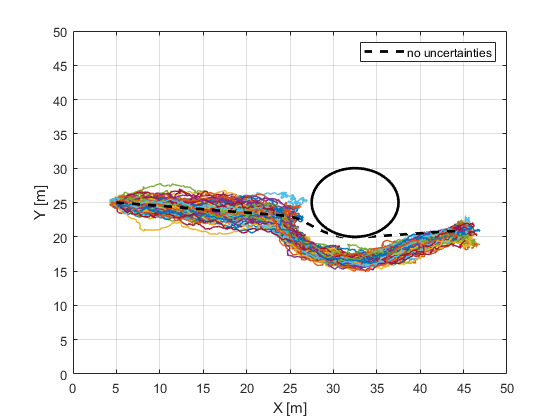}
     \caption{Trajectories obtained by running 100 randomized simulations. The dashed line indicated the trajectory obtained using a vanilla HOCBF when there are no uncertainties in the system.} 
    \label{fig: 100_runs_traj}
\end{figure}

\begin{figure}[t]
    \centering
    \includegraphics[width=0.75\columnwidth]{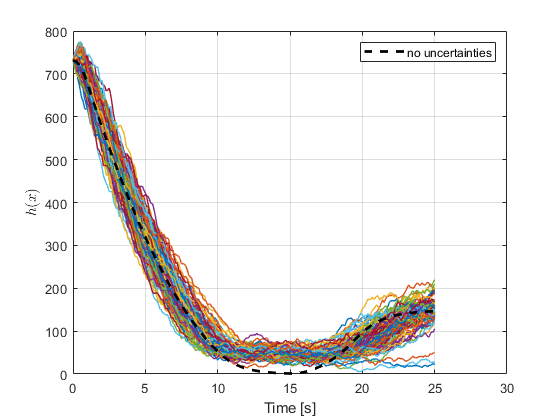}
    \caption{The value of $h(\mathbf{x})$ for 100 randomized runs. The dashed line indicated the trajectory obtained using a vanilla HOCBF when there are no uncertainties in the system. A negative value indicates a collision.}
    \label{fig: 100_runs_hx}
\end{figure}

\begin{table}[t]
\centering
\caption{A comparison of the least value of $h(\mathbf{x})$ obtained in the 100 runs to that obtained when there were no uncertainties in the system (considered  the benchmark).}
\begin{tabular}{|l|l|l|l|l|}
\hline
                & minimum & maximum & average & \begin{tabular}[c]{@{}l@{}}no \\ uncertainties\end{tabular} \\ \hline
$\min{h(\mathbf{x})}$ & 12.69   & 49.04   & 33.77   & 1.17                                                        \\ \hline
\end{tabular}
\label{table: multiple_run_comp}
\end{table}

\begin{figure*}
    \centering
    \includegraphics[width=\textwidth]{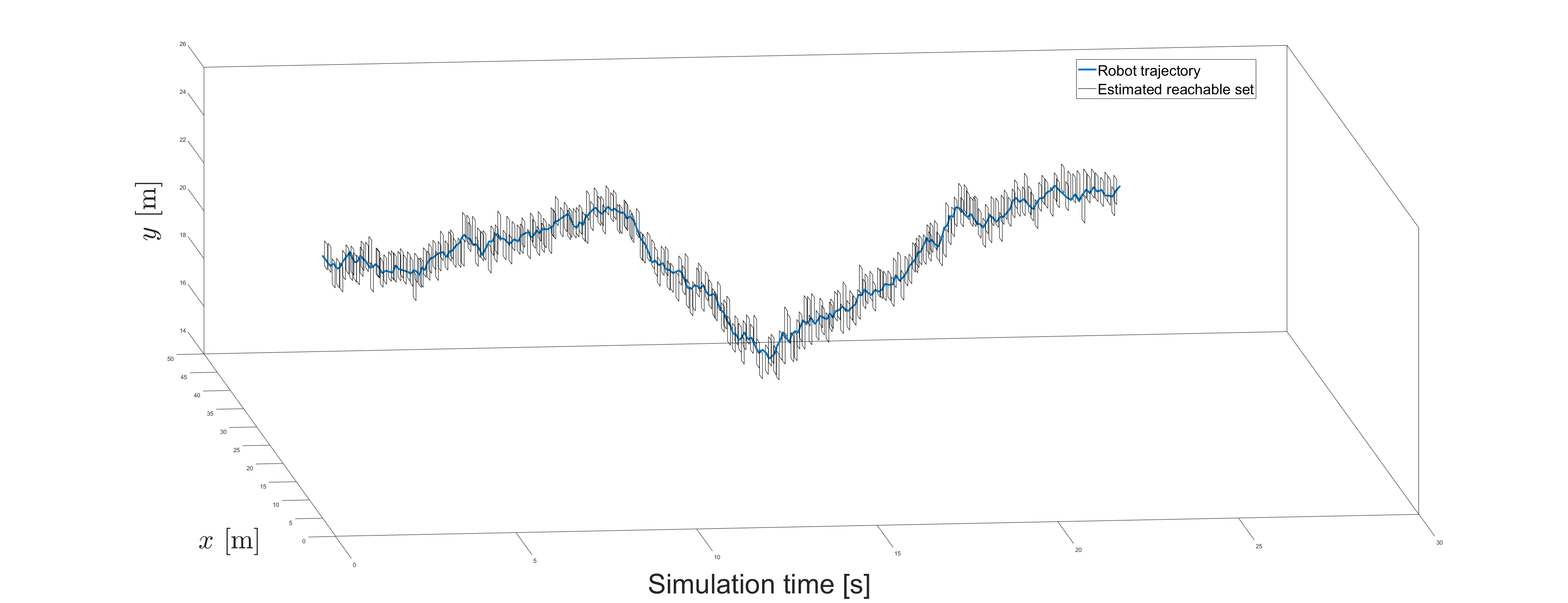}
    \caption{Reachable sets (black boxes) estimated at each time-step. The blue curve indicates the resulting trajectory taken by the robot.}
    \label{fig: reach_traj}
\end{figure*}

The robot starts from an initial position of $(5 ,25)$ and the goal is to reach $(45,21)$ while avoiding an obstacle at $(32.5, 25)$ by $D$. The safety-agnostic performance controller $K_{\text{perf}}(\hat{\mathbf{x}})$ is an MPC based controller. The resulting robot trajectories using the vanilla HOCBF approach and the proposed robust variation in \eqref{eq: robot_optim} are shown in Fig. \ref{fig: robot_trajs}. The safety requirement $h(\mathbf{x})$ is shown in Fig. \ref{fig: h_x}. The reachable sets estimated at each time-step are shown in Fig. \ref{fig: reach_traj}. It can be seen that the true trajectory taken by the robot passes through these sets. It can be clearly seen that the robot using the vanilla HOCBF (black curve) approach collides with the obstacle while the one using the proposed approach (blue curve) avoids collision.

Considering the stochasticity in the system, we perform 100 runs of the numerical problem described in the previous section, each time randomly sampling the disturbance error and measurement noise from their respective distributions. The resulting trajectories and the value of the safety requirement are shown in Figs. \ref{fig: 100_runs_traj} and \ref{fig: 100_runs_hx}, respectively, and are compared against the results obtained by implementing the vanilla HOCBF when there are no uncertainties in the system (the dashed line in the same figures). A comparison of the value of $h(\mathbf{x})$ is given in Table \ref{table: multiple_run_comp}. It can be seen that 32.6$\textrm{m}$ of performance (i.e., available free space to the obstacle) is lost on average due to the uncertainties in the system. 

\textit{Note: }From Figs. \ref{fig: 100_runs_traj} and \ref{fig: 100_runs_hx} it can be seen that are three trajectories that behave slightly differently from the bulk of trajectories due to the specific sequence of disturbance. Specifically, the robot goes towards  the obstacle and then turns down. While it may seem that the robot is unable to reach its goal, on further increasing the simulation time, the robot does indeed reach its goal while avoiding the obstacle. The behaviour of these trajectories is shown separately in Fig. \ref{fig: weird_traj}.
\section{Conclusion}
\label{sec: conclusion}

In this paper, we proposed doubly robust control barrier functions (DRCBF) for sampled-data systems with piece-wise constant controllers under bounded disturbance and measurement errors. Our DRCBF is applicable to scenarios where the relative degree of the control barrier function is 1 and higher. Additionally, for the high-order case we consider both matched disturbance and a special case of mis-matched disturbance where the relative degree of the control barrier funcion with respect to the disturbance is one lower than the input relative degree. We further incorporated interval reachability techniques to reduce the conservative behavior of the proposed sfae controller. A future direction of research is considering online learning of the disturbance bounds, using disturbance observers or Gaussian processes.

\bibliographystyle{IEEEtran}
\bibliography{IEEEabrv,ref.bib}

\clearpage
\appendix
\subsection{Proof of Lemma \ref{lemma: state_bounds}}
\begin{proof}
Over $t \in [kT, (k+1)T)$, under a constant input $\mathbf{u}_k$, we note that the evolution of the true state is $\dot{\mathbf{x}} = f(\mathbf{x}) + g(\mathbf{x})\mathbf{u}_k + p(\mathbf{x})\mathbf{d}$ while $\hat{\mathbf{x}}_k$ is constant (i.e., $\dot{\hat{\mathbf{x}}}_k = 0$). Given $f$, $g$, and $p$ are locally Lipschitz continuous, Caratheodery's theorem \cite{lars2011nonlinear} ensures that the solution exists and is unique for a bounded input $\mathbf{u}_k$. The rest of the proof follows \cite[Thm.~3.4]{khalil2015nonlinear}. 

The solutions of the true system \eqref{eq: d_cont_sys} and sampled state for all $t \in [t_k, t_k+ T]$ are:
\begin{equation*}
    \begin{split}
        & \mathbf{x}(t) = \mathbf{x}(t_k) + \int_{t_k}^{t_k + T}  [f(\mathbf{x}(s)) + g(\mathbf{x}(s))\mathbf{u}_k + p(\mathbf{x}(s))\mathbf{d}] ds\\
        & \hat{\mathbf{x}}(t) = \hat{\mathbf{x}}_k 
    \end{split}
\end{equation*}
Subtracting both equations and taking norms yields
\begin{equation*}
    \begin{split}
        &\lVert \mathbf{x}(t) - \hat{\mathbf{x}}(t) \rVert \leq \lVert \mathbf{x}(t_k) - \hat{\mathbf{x}}_k \rVert + \\
        &\quad \int_{t_k}^{t_k + T}  \lVert f(\mathbf{x}(s)) + g(\mathbf{x}(s))\mathbf{u}_k + p(\mathbf{x}(s))\mathbf{d} \rVert ~ds.
    \end{split}
\end{equation*}

From Assumption \ref{asmp: bounded_disturbance}, we know $\lVert \mathbf{d} \rVert \leq \gamma$ and from Assumption \ref{asmp: bounded_meas_error}, $\lVert \mathbf{x}(t_k) - \hat{\mathbf{x}}_k \rVert < \epsilon(\mathbf{z})$. Additionally, since $\mathcal{U}$ is compact, there exists $u_{\textrm{max}}$ such that $\lVert \mathbf{u}_k \rVert \leq u_{\textrm{max}}$. This gives us:
\begin{equation*}
        \lVert \mathbf{x}(t) - \hat{\mathbf{x}}(t) \rVert \leq \epsilon(\mathbf{z}) + T(\Delta + \lVert L_ph(\mathbf{x})) \rVert \gamma)
\end{equation*}
and concludes the proof.
\end{proof}

\subsection{When $h$ is of relative degree 1}

\subsubsection{For continuous time systems} We first address the bounded disturbance $\mathbf{d}$ in the continuous system \eqref{eq: cl_estim_sys}. We introduce the following definition of a \emph{doubly robust} CBF (DMR-CBF), i.e., robust to both external disturbance and measurement errors. 

\begin{definition}[DMR-CBF]
\label{def: dmr_cbf}
Consider the closed-loop system \eqref{eq: cl_estim_sys} and a safe set $\mathcal{S}$ defined as \eqref{eq: safe_set}. The continuously differentiable function $h: \mathbb{R}^n \rightarrow \mathbb{R}$ is a \emph{doubly robust CBF (DMR-CBF)} if there exist a suitable class-$\mathcal{K}$ function $\alpha$ and parametric functions $a, b: \mathbb{R}^l \rightarrow \mathbb{R}_{\geq 0}$ such that there exists a control input $\mathbf{u} \in \mathcal{U}$ under which:
\begin{multline}
\label{eq: dmr_cbf_condition}
    \sup_{\mathbf{u} \in \mathcal{U}} \Big[ L_fh(\hat{\mathbf{x}}) + L_gh(\hat{\mathbf{x}})\mathbf{u} + \alpha(h(\hat{\mathbf{x}})) \\ - (a(\mathbf{z}) + b(\mathbf{z})\lVert \mathbf{u} \rVert)\epsilon(\mathbf{x}) - \lVert L_ph(\mathbf{x}) \rVert \gamma \Big] \geq 0~. 
\end{multline}
\end{definition}

\begin{lemma}
\label{lemma: distrb_cont_reg_cbf}
    Consider system \eqref{eq: cl_estim_sys} and a safe set $\mathcal{S} \subset \mathbb{R}^n$ defined by a continuously differentiable function $h$ as \eqref{eq: safe_set}, and let $h$ be a DMR-CBF as defined in Definition \ref{def: dmr_cbf}. Assume that the functions $L_fh$, $L_gh$, and $\alpha \circ h$ are Lipschitz continuous with Lipschitz constants $\bar{L}_{L_fh}$, $\bar{L}_{L_gh}$, and $\bar{L}_{\alpha \circ h} \in \mathbb{R}_{\geq 0}$. Define parametric functions $a(\mathbf{z})=\epsilon(\mathbf{z})(\bar{L}_{L_fh} + \bar{L}_{L_gh})$ and $b(\mathbf{z})=\epsilon(\mathbf{z})(\bar{L}_{\alpha \circ h})$. Provided Assumptions \ref{asmp: bounded_disturbance} and \ref{asmp: bounded_meas_error} hold, any locally Lipschitz continuous controller $\mathbf{u} \in K_{\textrm{dmr}}$ renders \eqref{eq: cl_estim_sys} forward invariant with respect to $\mathcal{S}$.
\end{lemma}
\begin{proof}
We begin by noting that to prove the forward invariance of $\mathcal{S}$, according to \cite{ames2016control} it is sufficient to show \eqref{eq: cbf-def}. Expanding the left-hand side of \eqref{eq: cbf-def}:
    \begin{multline*}
        \dot{h}(\mathbf{x}) + \alpha(h(\mathbf{x})) \\
        = L_fh(\mathbf{x}) + L_gh(\mathbf{x})\mathbf{u} + L_ph(\mathbf{x})\mathbf{d} + \alpha(h(\mathbf{x}))
    \end{multline*}
\begin{equation*}
    \begin{aligned}
        & \geq L_fh(\mathbf{x}) + L_gh(\mathbf{x})\mathbf{u} - \lVert L_ph(\mathbf{x}) \rVert \lVert \mathbf{d} \rVert + \alpha(h(\mathbf{x})) \\
        & \geq L_fh(\mathbf{x}) + L_gh(\mathbf{x})\mathbf{u} - \lVert L_ph(\mathbf{x}) \rVert \gamma + \alpha(h(\mathbf{x}))
    \end{aligned}
\end{equation*}
The first inequality comes from the fact that $L_ph(\mathbf{x})\mathbf{d} \geq - \lVert L_ph(\mathbf{x}) \rVert \lVert \mathbf{d} \rVert$ and the second inequality comes from Assumption \ref{asmp: bounded_disturbance}. Let $c(\mathbf{x}, \mathbf{u}) := L_fh(\mathbf{x}) + L_gh(\mathbf{x})\mathbf{u} + \alpha(h(\mathbf{x})) - \lVert L_ph(\mathbf{x}) \rVert \gamma$. Now to account for the measurement errors, it is sufficient to show that:
\begin{equation*}
    \inf_{\mathbf{x} \in \mathcal{X}(\mathbf{z})} c(\mathbf{x}, \mathbf{u}) \geq 0
\end{equation*}
where $\mathcal{X}(\mathbf{z}) := \{\mathbf{x} \in \mathbb{R}^n : \exists ~ \mathbf{e} \in \mathcal{E}(\mathbf{z}) \text{ s.t. } \hat{\mathbf{x}} = \mathbf{x} + \mathbf{e}(\mathbf{x})\}$, which represents all the \emph{actual} states the system may lie in given a measurement-estimate pair ($\mathbf{z}$, $\hat{\mathbf{x}}$). The above inequality aims to show that $\dot{h}(\mathbf{x}) + \alpha(h(\mathbf{x})) \geq c(\mathbf{x}) \geq  \inf_{\mathbf{x} \in \mathcal{X}(\mathbf{z})} c(\mathbf{x}, \mathbf{u}) \geq 0$, which will conclude the proof. The rest of the proof is similar to that of Thm.~2 in \cite{dean2021guaranteeing} and is omitted here for the sake of brevity.
\end{proof}

\subsubsection{For sampled-data systems}
In this section, we extend the DMR-CBF formulation proposed in the previous section to sampled-data systems.

\begin{definition}[SD-DMR-CBF]
\label{def: sd_dmr_cbf}
    Consider the closed-loop system \eqref{eq: cl_estim_sys} and a safe set $\mathcal{S}$ defined as \eqref{eq: safe_set}. Assume that the states are sampled at fixed time steps $T$, and the states and control inputs are constant between sample times. The continuously differentiable function $h: \mathbb{R}^n \rightarrow \mathbb{R}$ is a \emph{doubly robust} control barrier function for this sampled data system (\emph{SD-DMR-CBF}),  if there exists a suitable class-$\mathcal{K}$ function $\alpha$ such that there exists a control input $\mathbf{u} \in \mathcal{U}$ under which:
    \begin{multline}
    \label{eq: sd_dmr_condition}
        \sup_{\mathbf{u}_k \in \mathcal{U}} \Big [L_fh(\hat{\mathbf{x}}_k) + L_gh(\hat{\mathbf{x}}_k)\mathbf{u}_k + \alpha(h(\hat{\mathbf{x}}_k)) - \lVert L_ph(\mathbf{x}) \rVert \gamma \\ - (\bar{L}_{L_fh} + \bar{L}_{L_gh}\lVert \mathbf{u}_k \rVert + \bar{L}_{\alpha \circ h})v(\mathbf{z}) \Big] \geq 0~.
    \end{multline}
\end{definition}

Following Definition \ref{def: sd_dmr_cbf}, denote the set of control inputs satisfying \eqref{eq: sd_dmr_condition} by $K_{\textrm{sd-dmr}}$. 
\begin{lemma}
\label{lemma: sd_dmr_cbf}
    Consider the sampled-data system \eqref{eq: cl_estim_sys} with a constant sampling time $T$ and a safe set $\mathcal{S}$ defined as \eqref{eq: safe_set} and a SD-DMR-CBF function $h$ following Definition \ref{def: sd_dmr_cbf}. Assume that the functions $L_fh$, $L_gh$, $L_ph$, and $\alpha \circ h$ are Lipschitz continuous with Lipschitz constants $\bar{L}_{L_fh}$, $\bar{L}_{L_gh}$, $\bar{L}_{L_ph}$ and $\bar{L}_{\alpha \circ h} \in \mathbb{R}_{\geq 0}$. Provided Assumptions \ref{asmp: bounded_disturbance} and \ref{asmp: bounded_meas_error} hold, any locally Lipschitz continuous controller $\mathbf{u} \in K_{\textrm{sd-dmr}}$ renders the sampled data system \eqref{eq: cl_estim_sys} forward invariant with respect to $\mathcal{S}$.
\end{lemma}
\begin{proof}
We prove this lemma by following a procedure similar to \cite{cortez2019control}. First, from the proof of Lemma \ref{lemma: distrb_cont_reg_cbf}, we have
\begin{equation*}
    \begin{aligned}
        & \dot{h}(\mathbf{x}) + \alpha(h(\mathbf{x}))  \\
        & =  L_fh(\mathbf{x}) + L_gh(\mathbf{x})\mathbf{u}_k + L_ph(\mathbf{x})\mathbf{d} + \alpha(h(\mathbf{x})) \\
        & \geq  L_fh(\mathbf{x}) + L_gh(\mathbf{x})\mathbf{u}_k + \alpha(h(\mathbf{x})) - \lVert L_ph(\mathbf{x}) \rVert \gamma~.
    \end{aligned}
\end{equation*}
Given the sampled state estimate $\hat{\mathbf{x}}_k$ and the actual state $\mathbf{x}$, consider the following: 
\begin{equation*}
\begin{aligned}
    & L_fh(\mathbf{x}) + L_gh(\mathbf{x})\mathbf{u}_k + \alpha(h(\mathbf{x})) \\ 
    & = L_fh(\hat{\mathbf{x}}_k) + L_gh(\hat{\mathbf{x}}_k)\mathbf{u}_k + \alpha(h(\hat{\mathbf{x}}_k)) + \\ 
    & (L_fh(\mathbf{x}) - L_fh(\hat{\mathbf{x}}_k)) + (L_gh(\mathbf{x})-L_gh(\hat{\mathbf{x}}_k))\mathbf{u}_k + \\
    & (\alpha(h(\mathbf{x})) - \alpha(h(\hat{\mathbf{x}}_k))) \\[7pt]
    & \geq L_fh(\hat{\mathbf{x}}_k) + L_gh(\hat{\mathbf{x}}_k)\mathbf{u}_k + \alpha(h(\hat{\mathbf{x}}_k)) -\\ 
    & \lVert (L_fh(\mathbf{x}) - L_fh(\hat{\mathbf{x}}_k)) \rVert - \lVert (L_gh(\mathbf{x})-L_gh(\hat{\mathbf{x}}_k))\mathbf{u}_k \rVert - \\
    & \lVert (\alpha(h(\mathbf{x})) - \alpha(h(\hat{\mathbf{x}}_k))) \rVert \\[7pt] 
    & \geq L_fh(\hat{\mathbf{x}}_k) + L_gh(\hat{\mathbf{x}}_k)\mathbf{u}_k + \alpha(h(\hat{\mathbf{x}}_k)) - \\ 
    & (\bar{L}_{L_fh} + \bar{L}_{L_gh}\lVert \mathbf{u}_k \rVert + \bar{L}_{\alpha \circ h})\lVert \mathbf{x} - \hat{\mathbf{x}}_k \rVert \\[7pt]
    & \geq L_fh(\hat{\mathbf{x}}_k) + L_gh(\hat{\mathbf{x}}_k)\mathbf{u}_k + \alpha(h(\hat{\mathbf{x}}_k)) - \\ 
    & (\bar{L}_{L_fh} + \bar{L}_{L_gh}\lVert \mathbf{u}_k \rVert + \bar{L}_{\alpha \circ h})v(\mathbf{z})~.
\end{aligned}
\end{equation*}
Now, we have:
\begin{multline*}
         \dot{h}(\mathbf{x}) + \alpha(h(\mathbf{x})) \geq L_fh(\hat{\mathbf{x}}_k) + L_gh(\hat{\mathbf{x}}_k)\mathbf{u}_k + \alpha(h(\hat{\mathbf{x}}_k)) \\
         - \lVert L_ph(\mathbf{x}) \rVert \gamma - (\bar{L}_{L_fh} + \bar{L}_{L_gh}\lVert \mathbf{u}_k \rVert + \bar{L}_{\alpha \circ h})v(\mathbf{z})~.
\end{multline*}
Choosing a control input from $\mathbf{u}_k \in K_{\textrm{sd-dmr}}$, we get $\dot{h}(\mathbf{x}) + \alpha(h(\mathbf{x})) \geq 0$ which concludes the proof.
\end{proof}

\subsection{When $h$ is of relative degree $m \geq 1$}
\subsubsection{Proof of Lemma \ref{lemma: ho_dmr_cbf_ct}}
From Theorem \ref{th: ho_cbf}, it is sufficient to show that $\psi_m \geq 0$. Following similar steps as in Lemma \ref{lemma: distrb_cont_reg_cbf}, we first obtain:
\begin{equation*}
    \begin{split}
    \psi_m \geq L_f\psi_{m-1}(\mathbf{x}) + L_g\psi_{m-1}(\mathbf{x})\mathbf{u} +  \alpha_{m-1}(\psi_{m-1}(\mathbf{x})) - \\ 
    \lVert L_p\psi_{m-1}(\mathbf{x}) \rVert \gamma
    \end{split}
\end{equation*}
The rest of the proof is similar to \cite[Prop. 1]{oruganti2023safe}.

\subsubsection{Proof for Lemma \ref{lemma: ho_dmr_cbf_sd}}
The proof is similar to Lemma \ref{lemma: sd_dmr_cbf} with the difference being that we look to prove $\psi_r \geq 0$.

\begin{figure}[t]
    \centering
    \includegraphics[width=\columnwidth]{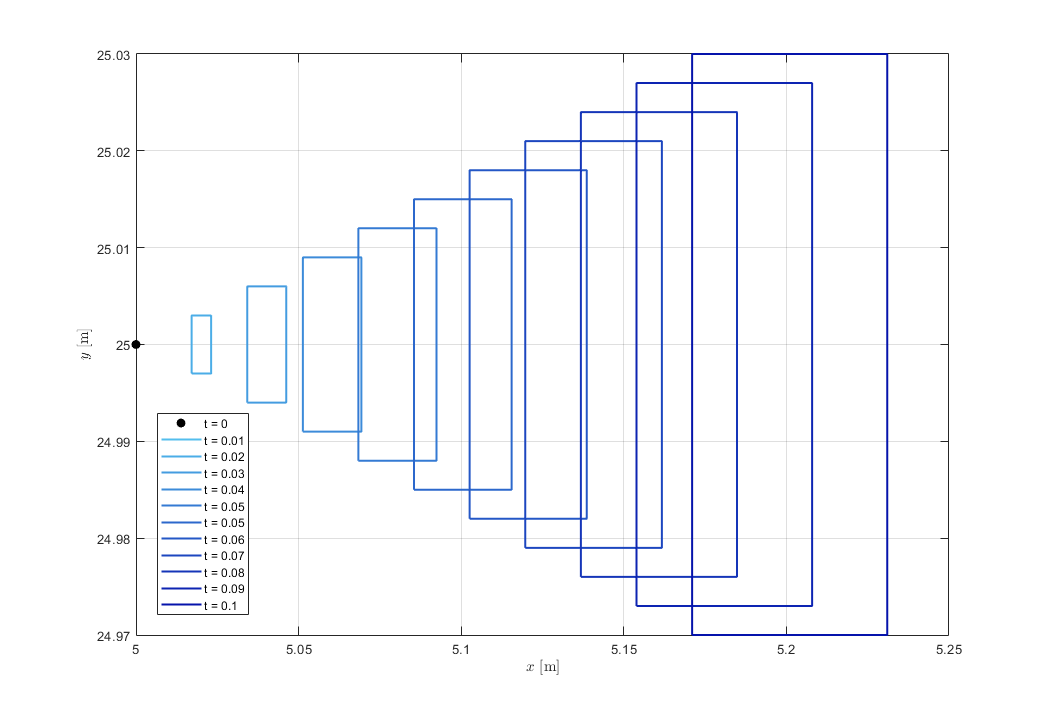}
    \caption{Reachable sets $\mathcal{R}(x_0, t_s)$ at different sampling intervals $t_s$ from an initial position of $x_0 = (5, 25)$}
    \label{fig: reach_sets}
\end{figure}

\begin{table}[t]
\centering
\caption{Robust margins for reachable sets estimated between sampling times $0.01$ to $0.1$. The size of the margins increases with increasing sampling time.}
\begin{tabular}{|l|l|}
\hline
Time & Margin  \\ \hline
0.01 & -31.980 \\ \hline
0.02 & -32.450 \\ \hline
0.03 & -32.735 \\ \hline
0.04 & -33.015 \\ \hline
0.05 & -33.290 \\ \hline
0.06 & -33.561 \\ \hline
0.07 & -33.087 \\ \hline
0.08 & -34.342 \\ \hline
0.09 & -34.593 \\ \hline
0.1  & -34.839 \\ \hline
\end{tabular}
\label{table: margins}
\end{table}

\subsection{Obtaining the margin in the numerical example}
Following the discussion provided in Section \ref{sec: mismatached_disturbance}, choosing $p_1 = q_1 = p_2 = q_2 = 1$ we have:
\begin{equation*}
    \begin{split}
        \psi_0(\mathbf{x}) & = h(\mathbf{x}) = (x-x_o)^2 + (y-y_o)^2 - D^2 \\[7pt]
        \psi_1(\mathbf{x}) & = \dot{\psi_0}(\mathbf{x}) + \psi_0(\mathbf{x}) \\
        & = L_f\psi_0(\mathbf{x}) + L_p\psi_0(\mathbf{x})\mathbf{d} + h(\mathbf{x}) \\
        & = (2v\cos\theta(x-x_o) + 2v\sin\theta(y-y_o)) + \\
        & (2d_1(x-x_o) + 2d_2(y-y_o)) + h(\mathbf{x}) \\[7pt]
        \psi_2(\mathbf{x}, \mathbf{d}) & = \dot{\psi_1}(\mathbf{x}, \mathbf{d}) + \psi_1(\mathbf{x}, \mathbf{d}) \\
        & \text{dropping dependence on $\mathbf{x}$ and $\mathbf{d}$ for readability}\\
        \psi_2 & = L_f\psi_1 + L_g\psi_1\mathbf{u} + L_p\psi_1\mathbf{d} + \psi_1 \\[7pt]
        & = \underbrace{L^2_fh + 2L_fh + 2h + L_gL_fh\mathbf{u}}_{\text{vanilla HOCBF}} + \\
        & \underbrace{L_pL_fh\mathbf{d} + L_f[L_ph\mathbf{d}] + L_p[L_ph\mathbf{d}]\mathbf{d} + 2L_ph\mathbf{d}}_{\text{effect of disturbance}} \\[7pt]
        & = (2v\cos\theta + 2d_1)v\cos\theta + (2v\sin\theta + 2d_2)v\sin\theta + \\
        & (2v\cos\theta(y-y_o) - 2v\sin\theta(x-x_o))u_1 + \\
        & (2(x-x_o)\cos\theta + 2(y-y_o)\sin\theta)u_2 + \\
        & (2v\cos\theta + 2d_1)d_1 + (2v\sin\theta + 2d_2)d_2 + \psi_1
    \end{split}
\end{equation*}
The $\mathtt{margin}$  function is:
\begin{equation*}
    \begin{split}
        \mathtt{margin} & = f_v(\mathbf{x}, \mathbf{u}) - f_v(\hat{\mathbf{x}}, \mathbf{u}) + f_d(\mathbf{x}, \mathbf{d})
    \end{split}
\end{equation*}
where \[f_v(\mathbf{x}, \mathbf{u}) = L^2_fh(\mathbf{x}) + 2L_fh(\mathbf{x}) + 2h(\mathbf{x}) + L_gL_fh(\mathbf{x})\mathbf{u} \] and 
\begin{equation*}
    \begin{split}
        f_d(\mathbf{x}, \mathbf{d}) & = L_pL_fh(\mathbf{x})\mathbf{d} + L_f[L_ph(\mathbf{x})\mathbf{d}] + \\
        & L_p[L_ph(\mathbf{x})\mathbf{d}]\mathbf{d} + 2L_ph(\mathbf{x})\mathbf{d}
    \end{split}
\end{equation*}

We note that the reachable sets obtained through the toolbox are the reachable sets \emph{at} time instant $t = (k+1)T$ whereas for the current approach we require a representation of all the states reachable \emph{through} time $t = kT$ to $t=(k+1)T$ (generally termed as the \emph{reachable tube}). But we observe numerically that for the current system, the margins estimated from the reachable set at time $(k+1)T$ from time $kT$ are larger than the margins at any time $t \in [kT, (k+1)T)$. A illustration of the reachable sets estimated at different sampling times is provided in Fig. \ref{fig: reach_sets}.

The robust margins obtained from these estimated reachable sets at every $0.01\textrm{s}$ sampling time $t_s$ from $0.01$ to $0.1$ are shown in Table \ref{table: margins}. Observe that the robust margin increases with increasing sampling time. A larger margin ensures that the system is still safe (albeit a bit conservatively) even between sampling times. Hence, we choose the robust margin obtained from the reachable set over-approximation estimated at time $T=0.1$.

\subsection{Trajectories that did not reach the goal}

\begin{figure}[h!]
    \centering
    \includegraphics[width=\columnwidth]{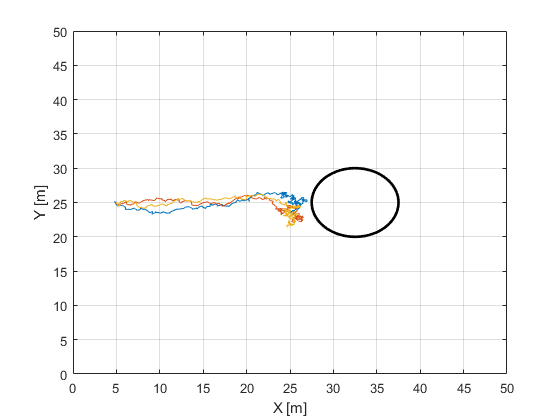}
    \caption{Trajectories that could not reach the goal within 25 seconds.}
    \label{fig: weird_traj}
\end{figure}

\end{document}